\documentclass[11pt,a4paper]{article}

\usepackage{mathtools}
\usepackage{authblk} 
\usepackage{algpseudocode,algorithmicx,algorithm}

\usepackage{mathrsfs}
\usepackage{latexsym,bm}
\usepackage{amsmath,amsfonts,amssymb,amsthm,rotating}
\usepackage{extarrows}

\usepackage{graphicx,subfigure,epstopdf,float}
\usepackage{enumerate,cases,multirow}
\usepackage{makecell}
\usepackage{caption}

\usepackage{longtable,colortbl,arydshln,threeparttable}
\definecolor{mygray}{gray}{.9}

\usepackage{indentfirst}
\usepackage[top=25mm,bottom=20mm,left=25mm,right=20mm]{geometry}
\usepackage{cite}

\usepackage{listings}

\usepackage{makeidx}        
\usepackage{booktabs}
\usepackage[bookmarksnumbered,colorlinks=true,citecolor=red,linkcolor=red,hyperindex,linktocpage=true]{hyperref}

\def \Arg {\mathrm{Arg}}

\newcommand{\polylog}{\mathrm{polylog}}

\newcounter{parentalgorithm}

\makeatother

\newtheorem{theorem}{Theorem}[section]
\newtheorem{lemma}{Lemma}[section]

\theoremstyle{remark}
\newtheorem{remark}{\bf Remark}[section]

\numberwithin{equation}{section}

\usepackage{qcircuit}

\begin{document}

		\title{Hybrid quantum-classical  algorithms for complex nonlinear partial differential equations with Ginzburg-Landau potential and vortex motion laws}
		\author[1]{Shi Jin\thanks{shijin-m@sjtu.edu.cn}}
		\author[1, 2]{Nana Liu\thanks{nana.liu@quantumlah.org} }
		\author[3]{Chuwen Ma\thanks{cwma@math.ecnu.edu.cn}}
		\affil[1]{School of Mathematical Sciences, Institute of Natural Sciences, MOE-LSC, Shanghai Jiao Tong University, Shanghai, 200240, China}
		\affil[2]{Global College, Shanghai Jiao Tong University, Shanghai 200240, China}
		\affil[3]{School of Mathematical Sciences, Key Laboratory of MEA, Ministry of Education, Shanghai Key Laboratory of PMMP, East China Normal University, Shanghai 200241, China
		}
		
			\maketitle
	
		\begin{abstract}
			We propose quantum algorithms for complex-valued  nonlinear  partial differential equations
			in the strongly nonlinear  regime, where the dynamics is governed by vortex
			cores, phase singularities, and nonlinear vortex interactions. 
            Examples include the complex-valued nonlinear Schr\"odinger equation, as well as nonlinear heat and wave equations with Ginzburg--Landau-type nonlinearity.
            In the strongly nonlinear regime, the solutions to these equations are asymptotically governed by, in leading order,  linear elliptic equations, coupled with low-dimensional vortex dynamics, where the vortex cores correspond to topological defects in superconductors. Our hybrid quantum-classical  algorithms utilize this asymptotic property, in which the vortex dynamic is advanced  classically while the boundary-value problem of linear elliptic equation is handled by quantum algorithms.   
            For the two-dimensional nonlinear Schr\"odinger equation, we also combine quantum BPX
            preconditioning with Schr\"odingerization to estimate physically relevant observables in the small-output regime. This yields, already in two dimensions, an {\it exponential} improvement in the dependence on the spatial problem size, while the dependence on the target accuracy remains essentially linear up to polylogarithmic factors.
            We further show that the same
			principle extends to dissipative Ginzburg--Landau vortex dynamics and to vortex
			filaments in three-dimensional superconductivity. Numerical results support the validity of this PDE reduction and the effectiveness of the proposed approach.
		\end{abstract}

		 \tableofcontents

\section{Introduction}

Quantum algorithms for differential equations have attracted significant
attention in recent years because they offer a possible way to mitigate the
curse of dimensionality for high-dimensional linear problems. In particular,
for linear systems, linear ordinary differential equations (ODEs), and linear
partial differential equations (PDEs) arising from discretization, a broad
literature has been developed around quantum linear-system solvers,
Hamiltonian simulation, and related linear-algebraic primitives, including the
HHL method, LCU-based improvements, QSVT/QSP techniques, spectral and
time-discretization-based ODE solvers, and quantum preconditioning strategies;
see, for example,
\cite{Harrow2009Quantum,Childs2017LCU,LT20,gilyen2019quantum,Wossnig2018,Clader2013,ShaoXiang2018,Berry-2014,BerryChilds2017ODE,Childs-Liu-2020,JinLiuYu2022QDM}.
For non-unitary linear dynamics, several unitarization strategies have also
been proposed, including LCHS-type methods, dilation approaches, and
Schr\"odingerization
\cite{ALL2023LCH,ACL2023LCH2,JLY22SchrLong,Schrshort,DLX25}. In particular,
Schr\"odingerization provides a simple and general framework for lifting linear
non-unitary dynamics to a higher-dimensional Schr\"odinger-type system, and has
been extended to a range of linear PDEs
\cite{JLLY23ABC,JLC23Maxwell,MJLWZ24,JLM24SchrBackward,JinLiu-LA,analogPDE,JL24JaynesCummings}.

For genuinely nonlinear PDEs, however, the situation is much less understood.
A natural strategy is to search for a reformulation in which the nonlinear
problem becomes linear while preserving the relevant computational advantage.
In some special classes of equations, such exact or time-discretely approximate
transformations are indeed available. Examples include the Cole--Hopf transform
for viscous Burgers-type equations and level-set-based linear representations
for certain Hamilton--Jacobi and hyperbolic equations
\cite{JinOsher2003,JinLiu2024NonlinearPDE,JinLiu2025ViscosityHJ}. But such
structures are rare, highly problem-dependent, and do not provide a general
route for nonlinear PDEs. More generally, exact linear representations for
nonlinear ODEs, such as Koopman--von Neumann type formulations, become much
less effective for PDEs, since one must first discretize in space, producing
very large nonlinear ODE systems whose dimension depends on the numerical
resolution \cite{Joseph2020,JinLiu2022Observables} thus lack quantum advantages.

Another broad class of approaches is based on approximation, embedding, or
operator lifting. This includes local linearization, weakly nonlinear
expansions, Carleman-type embeddings, and related observable-based,
Koopman-type, or linear-operator formulations
\cite{JinLiu2022Observables,LiuKoldenKrovi2021,LiuAnFang2023,JinLiu2025ViscosityHJ,Bravyi2025NoisyNonlinear}.
These methods can be useful in perturbative regimes, but they are generally
restricted to weak nonlinearities, polynomial-type nonlinear structure, or
short times, and may fail to capture the genuinely nonlinear phenomena that
develop in many PDEs, such as singularities, shocks, or other coherent
structures.
From the perspective of the present work, the main issue is not simply the size
of the nonlinear term, but the fact that the essential dynamics is often
carried by singular geometric objects that are not naturally revealed by
perturbative linearization.

A different line of work reformulates the governing equations in alternative
physical variables that are more compatible with quantum evolution, for
example Madelung-type or hydrodynamic Schr\"odinger-type descriptions for
fluid dynamics \cite{MengYang2023HSE}. Such formulations may be meaningful in
specific contexts, but they differ from the present work in that they
redesign the full nonlinear dynamics, rather than isolating the component most
advantageous for quantum computation.

The present paper follows a different route. Rather than linearizing the full
nonlinear field equation, we utilize the rigorous asymptotic reduction in the
strongly nonlinear regime. Examples of such problems include the complex-valued nonlinear Schr\"odinger equation, and the nonlinear heat and wave equations, with nonlinear reaction term governed by the Ginzburg-Landau potential. Here the nonlinear field dynamics can be asymptotically decomposed into two
components of very different nature:  a high-dimensional but linear boundary-value elliptic (or parabolic for nonlinear wave equations) problem coupled with a low-dimensional evolution of vortex
dynamics. In the two-dimensional setting,
the reduced vortex motion is governed by the classical point-vortex law in the
nonlinear Schr\"odinger case and by a gradient-flow law in the dissipative
Ginzburg--Landau case \cite{LinXin1998,Lin1999,Neu1990,E1994}. This
linear-nonlinear  asymptotic decomposition provides a natural entry point for quantum
algorithms: the singular nonlinear structure is retained in a small classical
system, whereas the genuinely high-dimensional part is linear and thus suitable
for quantum treatment.

Based on this idea, we formulate a hybrid quantum-classical framework. The vortex locations are advanced classically, whereas the linear elliptic subproblem is discretized into a sparse linear system and treated through BPX preconditioning and Schr\"odingerization-based quantum linear solvers. Since our goal is not to reconstruct the full outer field, but only to estimate a small number of physically relevant observables, the natural computational regime is the small-output regime. In this setting, for the two-dimensional nonlinear Schr\"odinger equation, the resulting complexity exhibits an {\it exponential} improvement in the dependence on the spatial problem size, while the dependence on the target accuracy remains essentially linear up to polylogarithmic factors. In particular, this advantage already appears in the two-dimensional setting considered in the present paper.

This approach should also be contrasted with the two classes of methods mentioned above. In the present work, we do not attempt to encode the entire nonlinear PDE into a larger linear/operator framework, nor do we seek a direct alternative field representation for the full nonlinear dynamics. Instead, we use the vortex-regime asymptotic reduction to split the problem into a low-dimensional nonlinear vortex subsystem and a high-dimensional linear subproblem of elliptic type. This distinction is important from the computational viewpoint: one does not gain much by simulating the reduced low-dimensional nonlinear subsystem on a quantum device. Here, the quantum computer is used only for the large, higher-dimensional linear subproblem.

The main contributions of this work are as follows. First, based on the asymptotic theory in the strongly nonlinear regime, we develop a hybrid quantum-classical framework for the two-dimensional nonlinear Schr\"odinger equation by combining classical vortex evolution with quantum treatment of the harmonic correction through BPX preconditioning and Schr\"odingerization; this leads, in the small-output regime, to an {\it exponential} improvement in the dependence on the spatial problem size, while the dependence on the target accuracy remains essentially linear up to polylogarithmic factors. Second, we show that the same approach applies beyond the nonlinear Schr\"odinger setting, including dissipative Ginzburg--Landau vortex dynamics and filamentary models in three-dimensional superconductivity, and we support this framework with numerical experiments.

The rest of the paper is organized as follows. In Section~\ref{sec:reduction}, we review the asymptotic behavior of the nonlinear Schr\"odinger equation and derive the corresponding linear harmonic correction problem together with the reduced vortex motion laws. Section~\ref{sec:schrodingerization} recalls the Schr\"odingerization framework for linear subroutines. In Section~\ref{sec:hybrid quantum-classical formulation}, we present the hybrid quantum-classical formulation for the asymptotic model. Section~\ref{sec:GL and 3D extensions} extends the same idea to Ginzburg--Landau vortex dynamics and to three-dimensional filamentary problems in superconductivity. Section~\ref{sec:computational perspectives} discusses the computational implications of the framework, including the resulting exponential improvement in the dependence on the spatial problem size and the essentially linear dependence on the target accuracy up to polylogarithmic factors in the small-output regime, and Section~\ref{subsec:numerical-verification} reports numerical tests supporting the reduction.

		\section{Asymptotic behavior with vortex dynamics}
		\label{sec:reduction}
		
		In this section, we explain how the nonlinear field dynamics is asymptotically reduced  to a
		coupled system consisting of: (i) a low-dimensional evolution of vortex
		locations, and (ii) a linear boundary-value problem for a smooth correction
		field associated with a given vortex configuration. This decomposition is the
		basic analytical structure behind our later hybrid quantum-classical algorithm.
		
		\subsection{The model and the asymptotic vortex ansatz}
		\label{sec:model}
		We consider the two-dimensional nonlinear Schr\"odinger equation
		\begin{equation}\label{eq:NLS-2D}
			\left\{
			\begin{aligned}
				i\partial_t u^\varepsilon
				&=-
				\Delta u^\varepsilon+\frac{1}{\varepsilon^2}(1-|u^\varepsilon|^2)u^\varepsilon
				&&\text{in } \Omega\times(0,T),\\
				u^\varepsilon&=g
				&&\text{on } \partial\Omega\times(0,T),
			\end{aligned}
			\right.
		\end{equation}
		defined in a bounded, simply connected domain \(\Omega\subset\mathbb R^2\), where
		\(g:\partial\Omega\to S^1\) is a smooth boundary datum. The associated
		Ginzburg--Landau energy is
		\begin{equation}\label{eq:GL-energy}
			E^\varepsilon(u^\varepsilon)
			=
			\int_\Omega
			\left(
			\frac12 |\nabla u^\varepsilon|^2
			+
			\frac{(1-|u^\varepsilon|^2)^2}{4\varepsilon^2}
			\right)\,d\bm x.
		\end{equation}
		
		Throughout this section, we work under the following assumptions.
		\begin{itemize}
			\item[(\textbf{H1})]
			The boundary datum satisfies
			$
			\deg(g,\partial\Omega)=M>0.
			$
			
			\item[(\textbf{H2})]
			restrict to  the same-sign unit-vortex regime, namely, the
			limiting configuration consists of \(M\) distinct vortices
			\[
			\bm a(t)=(\bm a_1(t),\dots,\bm a_M(t))\subset\Omega,
			\]
			and each vortex has degree \(+1\).
			
			\item[(\textbf{H3})]
			The solution lies in the vortex regime in the sense that
			\begin{equation}\label{eq:vortex-regime-assumption}
				E^\varepsilon(u^\varepsilon)\le M\pi\log\frac1\varepsilon+C_0,
			\end{equation}
			and the Jacobian concentrates at the distinct vortex points
			\(\bm a_1(t),\dots,\bm a_M(t)\).
		\end{itemize}
		
		The above assumptions are adopted only to simplify notation and exposition.
		The general theory allows mixed-sign and higher-degree vortices; we refer to
		\cite{Lin1999} for the corresponding rigorous statements. These generalizations
		do not affect the main mechanism emphasized in this paper, namely, the reduction
		of the nonlinear field dynamics to a low-dimensional vortex system coupled with
		a linear correction problem.
		
		When $\varepsilon<<1$, the leading order asymptotics of solution to  \eqref{eq:NLS-2D} yields $(1-|u^\varepsilon|^2)u^\varepsilon=0$. Thus either $|u^\varepsilon|=1$, or $u^\epsilon=0$. In the first case  $u$ is a complex function of modulus $1$. The case $h^\varepsilon=0$ corresponds to the so-called {\it vortices}, which are the defects in the applications of superconductor modeling. Hence the solution to \eqref{eq:NLS-2D} is, asymptotically, the coupling of these two type of solutions.  Under assumptions \textbf{(H2)}--\textbf{(H3)}, assume that the vortices are  concentrated at
		the points \(\bm a_1(t),\dots,\bm a_M(t)\), and away from these vortex cores the
		modulus \(|u^\varepsilon|\) remains close to \(1\). Consequently, the leading-order
		behavior is governed by the phase. The singular part of the phase is determined by
		the vortex locations and is given by
		\begin{equation}\label{eq:Theta-a}
			\Theta_{\bm a}(\bm x)=\sum_{j=1}^M \arg(\bm x-\bm a_j).
		\end{equation}
		Equivalently,
		$
		e^{i\Theta_{\bm a}(\bm x)}
		=
		\prod_{j=1}^M \frac{\bm x-\bm a_j}{|\bm x-\bm a_j|}
		$
		away from the vortex cores.
		
		According to \cite{Lin1999}, on the time scale \(t=O(1)\), the linear momentum
		\(p(u^\varepsilon)=u^\varepsilon\wedge \nabla u^\varepsilon\) converges weakly,
		away from the vortex trajectories, to a vector field \(\bm v\) on
		\[
		\Omega_{\bm a}(t):=\Omega\setminus\{\bm a_1(t),\dots,\bm a_M(t)\}.
		\]
		The limiting field \(\bm v\) solves the incompressible Euler equations
		\[
		\bm v_t=2\,\bm v\cdot\nabla \bm v-2\nabla P,\qquad
		\operatorname{div}\bm v=0\quad\text{in }\Omega_{\bm a}(t),\qquad
		\bm v\cdot\tau=g\wedge g_\tau\quad\text{on }\partial\Omega,
		\]
		where \(P\) denotes the associated pressure, \(\tau\) is the unit tangential
		vector on \(\partial\Omega\), and
		\(g_\tau:=\nabla g\cdot\tau\) is the tangential derivative of \(g\). Moreover,
		\(\bm v\) admits the representation
		\begin{equation}\label{eq:vel}
			\bm v=\nabla(\Theta_{\bm a}+h_{\bm a}),
		\end{equation}
		where \(h_{\bm a}\) is a harmonic correction in \(\Omega\).
		
		Consistent with the limiting description above, as \(\varepsilon\to0\), the solution admits the asymptotic vortex ansatz
		\begin{equation}\label{eq:asymptotic-ansatz}
			u^\varepsilon \rightharpoonup
			u^0(\bm x)=u_{\bm a}(\bm x):=
			e^{i(\Theta_{\bm a}(\bm x)+h_{\bm a}(\bm x))}
			=
			e^{ih_{\bm a}(\bm x)}
			\prod_{j=1}^M \frac{\bm x-\bm a_j}{|\bm x-\bm a_j|},
		\end{equation}
		away from the vortex cores.

		\subsection{The harmonic correction associated with a vortex configuration}
		\label{sec:harmonic correction}
		Once the vortex locations \(\bm{a}=(\bm{a}_1,\dots,\bm{a}_M)\) are fixed, the singular phase \(\Theta_a\) is determined by \eqref{eq:Theta-a}. Hence, in \eqref{eq:asymptotic-ansatz}, the only remaining unknown is the harmonic correction \(h_a\).
		
		To see formally why \(h_{\bm a}\) is harmonic, recall from \eqref{eq:vel} that
		the limiting velocity field satisfies
		\[
		\bm v=\nabla(\Theta_{\bm a}+h_{\bm a}),
		\qquad
		\operatorname{div}\bm v=0
		\quad\text{in }\Omega_{\bm a}.
		\]
		Hence
		\[
		\Delta(\Theta_{\bm a}+h_{\bm a})=0
		\quad\text{in }\Omega_{\bm a}.
		\]
		Since \(\Theta_{\bm a}\) is harmonic away from the vortex points, namely
		$ \Delta\Theta_{\bm a}=0 $ in $\Omega_{\bm a}$,
		it follows formally that \(h_{\bm a}\) is harmonic in \(\Omega\).
		
		By \cite[Theorem 1.1]{Lin1999}, the limiting velocity field also satisfies
		\[
		\bm v\cdot\tau=g\wedge g_\tau
		\qquad\text{on }\partial\Omega.
		\]
		Taking the tangential component of
		\(\bm v=\nabla(\Theta_{\bm a}+h_{\bm a})\) on \(\partial\Omega\), one obtains
		\[
		\bm v\cdot\tau
		=
		\nabla(\Theta_{\bm a}+h_{\bm a})\cdot\tau
		=
		\Theta_{{\bm a},\tau}+h_{{\bm a},\tau}.
		\]
		Comparing this identity with the boundary condition above yields
		\begin{equation}\label{eq:ha-tangential-bc}
			h_{{\bm a},\tau}
			=
			-\Theta_{{\bm a},\tau}+g\wedge g_\tau
			\qquad\text{on }\partial\Omega,
		\end{equation}
		where \(h_{{\bm a},\tau}:=\nabla h_{\bm a}\cdot\tau\) and
		\(\Theta_{{\bm a},\tau}:=\nabla\Theta_{\bm a}\cdot\tau\).
		Thus \(h_{\bm a}\) is determined up to an additive constant.
		
		The function \(\Theta_{\bm a}\) carries the topological singularities generated
		by the vortices, whereas \(h_{\bm a}\) describes the regular contribution
		required to match the boundary condition. Equivalently, after fixing a boundary
		phase lifting \(\phi_g\) on \(\partial\Omega\) such that \(g=e^{i\phi_g}\),
		\eqref{eq:ha-tangential-bc} may be recast as the Dirichlet problem
		\begin{equation}\label{eq:ha-Dirichlet}
			\left\{
			\begin{aligned}
				\Delta h_{\bm a} &=0
				\qquad &&\text{in }\Omega,\\
				h_{\bm a}&=\phi_g-\Theta_{\bm a}
				\qquad &&\text{on }\partial\Omega,
			\end{aligned}
			\right.
		\end{equation}
		modulo the natural \(2\pi\)-ambiguity of the phase. This equivalent
		formulation is convenient for the discrete treatment in the next section,
		since it shows that once the vortex configuration is given, the regular
		correction \(h_{\bm a}\) is recovered from a standard boundary-value problem of the ({\it linear}) Laplace equation.

		\subsection{Renormalized energy and the point-vortex motion law}
		\label{sec:motion law}
		
		The effective interaction among the vortices is described by the renormalized
		energy. For the limiting configuration
		$
		u_{\bm a}(x)
		=
		e^{ih_{\bm a}(x)}
		\prod_{j=1}^M \frac{x-\bm a_j}{|x-\bm a_j|},
		$
		the renormalized energy \(W(\bm a)=W(\bm a_1,\dots,\bm a_M)\) is defined by
		\begin{equation}\label{eq:renormalized-energy}
			W(\bm a)
			=
			\lim_{r\downarrow 0}
			\left[
			\frac{1}{2\pi}
			\int_{\Omega\setminus \bigcup_{j=1}^M B_r(\bm a_j)}
			|\nabla u_{\bm a}|^2\,d\bm x
			-
			M\log\frac1r
			\right]
			+\gamma M,
		\end{equation}
		where \(B_r(\bm a_j):=\{\bm x\in\mathbb R^2:\ |\bm x-\bm a_j|<r\}\) denotes the open disk
		of radius \(r\) centered at \(\bm a_j\), and \(\gamma\) is a universal constant.
		This is the finite part of the energy
		after subtracting the singular core contribution of the vortices.
		
		To obtain a closed reduced dynamics, we further restrict attention to the
		almost-minimizing regime
		\begin{equation}\label{eq:almost-minimizing}
			E^\varepsilon(u^\varepsilon)(0)
			=
			M\pi\log\frac1\varepsilon+\pi W(\bm a(0))+o(1).
		\end{equation}
		Under this assumption, the defect measure vanishes and the vortex motion is
		completely determined by the renormalized energy. More precisely, by
		\cite[Theorem~1.2]{Lin1999}, if \(H_j\) denotes the smooth part of
		\(\Theta_{\bm a}+h_{\bm a}\) near the \(j\)-th vortex, then
		\begin{equation}\label{eq:rigorous-vortex-law}
			\dot{\bm a}_j
			=
			J\nabla_{\bm a_j}W(\bm a)
			=
			-2\nabla H_j(\bm a_j), \qquad
			J=
			\begin{pmatrix}
				0&-1\\
				1&0
			\end{pmatrix},
			\qquad j=1,\dots,M.
		\end{equation}
		
		For the present same-sign unit-vortex setting, this law can be written in a
		directly computable form. Indeed, near \(\bm x=\bm a_j\), the smooth part is
		\(
		H_j(\bm x)=\sum_{k\neq j}\arg(\bm x-\bm a_k)+h_{\bm a}(\bm x)
		\),
		so using
		\(
		\nabla\arg(\bm x-\bm a_k)=\frac{(\bm x-\bm a_k)^\perp}{|\bm x-\bm a_k|^2}
		\)
		with
		\(
		\bm x^\perp:=(-x_2,x_1)
		\),
		we obtain
		\begin{equation}\label{eq:M2-model}
			\dot{\bm a}_j
			=
			-2\left(
			\sum_{k\neq j}\frac{(\bm a_j-\bm a_k)^\perp}{|\bm a_j-\bm a_k|^2}
			+
			\nabla h_{\bm a}(\bm a_j)
			\right),
			\qquad j=1,\dots,M.
		\end{equation}
		Hence the vortex dynamics consists of an explicit pairwise Kirchhoff interaction
		together with a boundary-induced drift encoded by the harmonic correction
		\(h_{\bm a}\).
		
		When the boundary effect is negligible over the time interval of interest, one
		may further adopt the free-space approximation
		\begin{equation}\label{eq:M1-model}
			\dot{\bm a}_j
			=
			-2\sum_{k\neq j}\frac{(\bm a_j-\bm a_k)^\perp}{|\bm a_j-\bm a_k|^2},
			\qquad j=1,\dots,M.
		\end{equation}
		This reduced law is consistent with the classical matched-asymptotic derivation
		of Neu~\cite{Neu1990} in the whole-plane setting. For bounded domains, however,
		the rigorous motion law remains \eqref{eq:rigorous-vortex-law}, or equivalently
		\eqref{eq:M2-model}.
        
      For later use, we refer to the fully coupled bounded-domain model
\eqref{eq:M2-model} as \emph{M2}, and to the free-space approximation
\eqref{eq:M1-model} as \emph{M1}.

		\begin{remark}
			The bounded-domain reduced vortex laws for the nonlinear Schr\"odinger equation
			have also been investigated numerically by \cite{BaoTang2014},
			who compared the reduced ODE dynamics with direct simulations of the full NLSE
			under both Dirichlet and homogeneous Neumann boundary conditions. Their study
			provides useful numerical evidence that the reduced laws capture the leading-order
			vortex motion when $\varepsilon$ is small, while deviations may arise when boundary
			effects, radiation, or vortex--sound interactions become significant. In the present
			work, we use the same reduced-law viewpoint for a different purpose: not to build a
			standalone classical vortex solver, but to expose the nonlinear-to-linear decomposition
			that makes the outer-field computation amenable to quantum treatment.
		\end{remark}
		
		\section{A brief review of Schr\"odingerization}
		\label{sec:schrodingerization}
		
		Schr\"odingerization provides a general mechanism for converting linear but
		typically non-unitary dynamics into a Schr\"odinger-type evolution that can be
		handled by quantum simulation. The central idea is to embed the original problem
		into one higher dimension through an auxiliary variable, so that the resulting
		evolution is generated by a Hermitian Hamiltonian. In this way, linear ODEs,
		PDEs, and linear-algebraic subroutines can be reformulated in a form compatible
		with unitary quantum computation.
		
		For the linear system
		\[
		K\bm x=\bm b,
		\]
		where \(K\in\mathbb R^{N_h\times N_h}\) is a symmetric positive definite matrix
		arising from spatial discretization, and \(N_h\) is the number of spatial
		degrees of freedom.
		We start from a convergent stationary iteration
		\[
		\bm x^{\,\mathrm{new}}
		=
		\bm x^{\,\mathrm{old}}+B(\bm b-K\bm x^{\,\mathrm{old}}),
		\]
		where \(B\) is an iterator, or more generally a preconditioner. The convergence
		of this iteration is equivalent to the relaxation of the linear ODE
		\[
		\frac{d\bm u}{dt}
		=
		-BK\bm u+B\bm b,
		\qquad
		\bm u(0)=\bm u_0,
		\]
		whose steady state is \(\bm u_\infty=\bm x\).
		
		For quantum computation, it is preferable to work with a symmetrically
		preconditioned form. Writing
		$
		B=SS^\top,
		$
		we introduce
		\[
		K_S:=S^\top K S,
		\qquad
		\bm b_S:=S^\top \bm b,
		\qquad
		\bm u(t)=S\bm z(t),
		\]
		so that the original problem is reduced to
		\[
		\frac{d\bm z}{dt}
		=
		-K_S\bm z+\bm b_S.
		\]
		This step is crucial: it replaces the generally nonsymmetric operator \(BK\) by
		the symmetric preconditioned operator \(K_S\), and it aligns naturally with the
		block-encoding framework used later for quantum implementation.
		
		To remove the inhomogeneous term, we augment the system from
		\cite{JinLiuMa2025} and write
		\[
		\bm z_f(t)
		=
		\begin{bmatrix}
			\bm z(t)\\
			T\bm b_S
		\end{bmatrix},
		\qquad
		\frac{d\bm z_f}{dt}
		=
		K_f\bm z_f,
		\qquad
		K_f=
		\begin{bmatrix}
			-K_S & I/T\\
			O & O
		\end{bmatrix},
		\]
		where
		$
		T=\Theta\!\left(\frac{1}{\lambda_{\min}(BK)}\log\frac{1}{\varepsilon}\right)
		$
		is chosen so that the relaxation reaches accuracy \(\varepsilon\), namely
		\[
		\|\bm x-\bm u(T)\|\le \varepsilon \|\bm x\|.
		\]
		The solution of the original linear system is encoded in the first block of
		\(\bm z_f(T)\).
		
		Schr\"odingerization is then carried out through a warped phase transform in an
		auxiliary variable \(p\). Introducing
		$\bm w(t,p)=e^{-p}\bm z_f(t) $, for $p>0$,
		and extending the equation naturally to the whole real line in \(p\), one
		obtains
		\[
		\partial_t \bm w
		=
		-H_1\,\partial_p\bm w+iH_2\bm w,
		\qquad
		\bm w(0,p)=\psi(p)\bm z_f(0),
		\]
		where \(\psi\) is chosen so that \(\psi(p)\approx e^{-p}\) for \(p\ge 0\); see
		\cite{JLMPY2026}. A typical choice is \(\psi(p)=e^{-|p|}\). Here
		$H_1=\frac{K_f+K_f^\dagger}{2}$ and $ H_2=\frac{K_f-K_f^\dagger}{2i} $
		are Hermitian. Thus the original dissipative linear problem is transformed into
		a Schr\"odinger-type equation in one higher dimension. Moreover, the original
		variable is recovered by
		\[
		\bm z_f(t)=e^{p}\bm w(t,p), \qquad
		p\ge p_3:=\lambda_{\max}^+(H_1)\,T=1/2.
		\]
		
		\subsection{Implementation of the Schr\"odingerized evolution}
		
		After truncating the auxiliary variable \(p\) to \([-R,R]\), we discretize the
		\(p\)-direction on the uniform grid
		\[
		p_k=-R+k\Delta p,\qquad k=0,\dots,N_p-1,\qquad \Delta p=2R/N_p,
		\]
		with \(N_p=2^{n_p}\). Applying a Fourier spectral discretization, the
		Schr\"odingerized system is discretized to
		\[
		\frac{d}{dt}\bm W_h(t)=-i\widetilde H\bm W_h(t),
		\qquad
		\widetilde H=(\Phi D_p\Phi^{-1})\otimes H_1-I\otimes H_2,
		\]
		where \(D_p=\operatorname{diag}(\mu_0,\dots,\mu_{N_p-1})\) collects the Fourier
		wave numbers in the auxiliary \(p\)-direction with
		$\mu_\ell=\frac{\pi(\ell-N_p/2)}{R}$, 
		and therefore represents the operator \(-i\partial_p\) after spectral
		discretization. Here
		\[
		\Phi=(e^{i\mu_\ell p_k})_{k,\ell=0}^{N_p-1}
		\]
		denotes the discrete Fourier transform matrix in the \(p\)-register. Hence
		\[
		\bm W_h(T)=e^{-i\widetilde HT}\bm W_h(0),\qquad
		\bm W_h(0)=\bm\psi_h\otimes \bm z_f(0),
		\]
		with
		$
		\bm\psi_h=(\psi(p_0),\dots,\psi(p_{N_p-1}))^\top.
		$
		
		Suppose that we are given the input-state preparation oracle
		\[
		O_{\mathrm{prep}}:\ |0\rangle|0\rangle \mapsto |\bm\psi_h\rangle|\bm z_f(0)\rangle,
		\]
		together with a block-encoding of the symmetrically preconditioned operator
		\(K_S=S^\top K S\), from which the Hamiltonian simulation oracle
		\[
		O_H(\tau):\ |0\rangle|\bm\varphi\rangle \mapsto |0\rangle e^{-iH\tau}|\bm\varphi\rangle,
		\qquad
		H=D_p\otimes H_1-I\otimes H_2
		\]
		for \(|\tau|\le T\) is constructed. Then the Schr\"odingerized evolution is
		implemented coherently as
		\[
		|0\rangle|0\rangle
		\xrightarrow{\ O_{\mathrm{prep}}\ }
		|\bm\psi_h\rangle|\bm z_f(0)\rangle
		\xrightarrow{\ \mathrm{QFT}\otimes I\ }
		(\Phi^{-1}|\bm\psi_h\rangle)\otimes |\bm z_f(0)\rangle
		\xrightarrow{\ O_H(T)\ }
		e^{-iHT}\bigl((\Phi^{-1}|\bm\psi_h\rangle)\otimes |\bm z_f(0)\rangle\bigr)
		\xrightarrow{\ \mathrm{QFT}\otimes I\ }
		|\bm W_h(T)\rangle .
		\]
		
		To recover the steady-state component, define the recovery projector
		\[
		\Pi_{\mathrm{rec}}
		=
		\sum_{k\in \mathcal I_3}|k\rangle\langle k|\otimes I,
		\qquad
		\mathcal I_3:=\{k:\ p_k\ge p_3,\ \ p_k=O(1)\}.
		\]
		Conditioned on measuring the auxiliary register in \(\mathcal I_3\), the
		post-measurement state yields
		\[
		|k\rangle\otimes |\bm z_f(T)\rangle
		\]
		up to the known factor \(e^{-p_k}\) introduced by the warped phase transform.
		Hence the first block of the recovered state gives an approximation to
		\(|\bm z(T)\rangle\), and therefore to
		\[
		|\bm h\rangle=|S\bm z(T)\rangle.
		\]
		
		We are not interested in reconstructing the full vector \(\bm h\). Instead, for
		a linear observable
		\[
		\mathcal O(\bm h)=\langle \bm c,\bm h\rangle,
		\]
		we write
		\[
		\mathcal O(\bm h)
		=
		\langle \bm c,S\bm z(T)\rangle
		=
		\langle S^\top \bm c,\bm z(T)\rangle.
		\]
		Let
		$|c_S\rangle := \frac{S^\top \bm c}{\|S^\top \bm c\|} $
		be the normalized measurement state, prepared by an oracle
		\[
		O_{\mathrm{meas}}:\ |0\rangle \mapsto |c_S\rangle .
		\]
		Then the desired quantity reduces to the overlap
		\[
		\mathcal O(\bm h)
		=
		\|S^\top \bm c\|\,\|\bm z(T)\|\,
		\langle c_S|z(T)\rangle,
		\]
		which can be estimated by a Hadamard test or amplitude estimation. Thus the
		quantum part of the algorithm consists of the coherent preparation of the input
		state, Hamiltonian simulation of \(e^{-iHT}\), postselected recovery of the
		steady-state component, and overlap estimation against the measurement state
		\(|c_S\rangle\).
		
		\begin{theorem}[Schr\"odingerization-based estimation of linear observables]
			\label{thm:complexity}
			Let \(K_S=S^\top K S\) be the symmetrically preconditioned operator, and assume
			that \(K_S\) admits a block-encoding with normalization factor
			\[
			\alpha_{K_S}=O(\operatorname{poly}(d)\log(N_h)).
			\]
			Assume moreover that
			$\lambda_{\min}(K_S)=\Theta(1) $, $\lambda_{\max}(K_S)=\Theta(1)$, 
			and that the state-preparation oracles \(O_{\mathrm{prep}}\) and
			\(O_{\mathrm{meas}}\) for the normalized states
			\(|\bm z_f(0)\rangle\) and \(|\bm c_S\rangle\), respectively, are
			available. Then, for any \(\epsilon\in(0,1)\), there exists a quantum
			algorithm which outputs an estimate \(\widetilde{\mathcal O}\) of the linear
			observable
			$
			\mathcal O(\bm h)=\langle \bm c,\bm h\rangle
			$
			such that
			\[
			|\widetilde{\mathcal O}-\mathcal O(\bm h)|\le \epsilon .
			\]
			Furthermore, in the preconditioned Poisson setting of \cite{JLMY25}, the query
			complexity for estimating \(\mathcal O(\bm h)\) is
			\[
			O\!\bigl(\operatorname{poly}(d)\polylog(N_h)\,\epsilon^{-1}
			\,\operatorname{polylog}(\epsilon^{-1})\bigr).
			\]
		\end{theorem}

		\section{A hybrid quantum-classical formulation}
		\label{sec:hybrid quantum-classical formulation}
		We now turn to the numerical treatment of the reduced model derived in
		Section~\ref{sec:reduction}. The key observation is that the reduced dynamics
		consists of two components of very different nature: a low-dimensional vortex
		system and a linear harmonic correction problem associated with the current
		vortex configuration. This naturally leads to a hybrid strategy: the vortex
		trajectories are advanced classically, whereas the harmonic correction is
		computed through a large linear system and treated by a quantum linear solver.
		
		\subsection{Classical time-marching for the vortex dynamics}
		\label{sec:classical time-stepping}
		
		We begin with the classical update of the vortex locations. As discussed in
		Section~\ref{sec:motion law}, the full bounded-domain reduced model is
		\begin{equation}\label{eq:M2-model-sec3}
			\dot{\bm a}_j
			=
			-2\left(
			\sum_{k\neq j}\frac{(\bm a_j-\bm a_k)^\perp}{|\bm a_j-\bm a_k|^2}
			+
			\nabla h_{\bm a}(\bm a_j)
			\right),
			\qquad j=1,\dots,M.
		\end{equation}
		Thus, in the fully coupled dynamics, the vortex update depends on the regular
		correction through the local values \(\nabla h_{\bm a}(\bm a_j)\).
		
		In the present work, however, our primary operating mode is the simplified
		free-space point-vortex model
		\begin{equation}\label{eq:M1-model-sec3}
			\dot{\bm a}_j
			=
			-2\sum_{k\neq j}\frac{(\bm a_j-\bm a_k)^\perp}{|\bm a_j-\bm a_k|^2},
			\qquad j=1,\dots,M,
		\end{equation}
		which is exactly the approximation introduced in
		Section~\ref{sec:motion law} when the boundary contribution is negligible over
		the time interval of interest. Under this approximation, the vortex motion is
		completely decoupled from the harmonic correction problem.
		
		Let \(0=t_0<t_1<\cdots<t_{N_T}=T\) be a uniform partition of the time interval
		with time step \(\Delta t=T/N_T\). Denoting by \(\bm a_j^m\) the numerical
		approximation of \(\bm a_j(t_m)\), a first-order explicit discretization of
		\eqref{eq:M1-model-sec3} is
		\begin{equation}\label{eq:classical-vortex-update}
			\bm a_j^{m+1}
			=
			\bm a_j^m
			-
			2\Delta t
			\sum_{k\neq j}
			\frac{(\bm a_j^m-\bm a_k^m)^\perp}{|\bm a_j^m-\bm a_k^m|^2},
			\qquad j=1,\dots,M.
		\end{equation}
		This yields a fully classical evolution of the vortex trajectories
		\[
		\bm a^m=(\bm a_1^m,\dots,\bm a_M^m),
		\qquad m=0,1,\dots,N_T.
		\]
		
		If one instead adopts the fully coupled model
		\eqref{eq:M2-model-sec3}, then the corresponding first-order explicit update is
		\begin{equation}\label{eq:classical-vortex-update-coupled}
			\bm a_j^{m+1}
			=
			\bm a_j^m
			-
			2\Delta t
			\left(
			\sum_{k\neq j}
			\frac{(\bm a_j^m-\bm a_k^m)^\perp}{|\bm a_j^m-\bm a_k^m|^2}
			+
			\nabla h_{\bm a^m}(\bm a_j^m)
			\right),
			\qquad j=1,\dots,M.
		\end{equation}
		Hence even in the coupled model, the classical update does not require the full
		field \(h_{\bm a^m}\), but only its gradients at the current vortex locations.
		
		Since the reduced free-space point-vortex system \eqref{eq:M1-model-sec3} is
		Hamiltonian, one may alternatively use a structure-preserving time integrator.
		A natural second-order symplectic choice is the implicit midpoint rule
		\begin{equation}\label{eq:classical-vortex-update-implicit-midpoint}
			\frac{\bm a_j^{m+1}-\bm a_j^m}{\Delta t}
			=
			-2\sum_{k\neq j}
			\frac{\left(\frac{\bm a_j^{m+1}+\bm a_j^m}{2}
				-\frac{\bm a_k^{m+1}+\bm a_k^m}{2}\right)^\perp}
			{\left|\frac{\bm a_j^{m+1}+\bm a_j^m}{2}
				-\frac{\bm a_k^{m+1}+\bm a_k^m}{2}\right|^2},
			\qquad j=1,\dots,M.
		\end{equation}
		Compared with the explicit scheme \eqref{eq:classical-vortex-update}, this
		midpoint discretization better respects the Hamiltonian structure and is
		preferable for long-time simulations. In the present paper, however, we use the
		explicit first-order scheme \eqref{eq:classical-vortex-update} for simplicity,
		and regard \eqref{eq:classical-vortex-update-coupled} as the corresponding
		first-order hybrid extension.

		\subsection{Quantum treatment of the harmonic correction problem}
		\label{subsec: quantum treatment}
		
		We now describe the quantum part of the hybrid framework. Once the vortex
		locations have been updated classically, the remaining task is to compute the
		harmonic correction associated with the current vortex configuration. The key
		point is that, for our purposes, the quantum routine is not used to reconstruct
		the full field \(h_{\bm a^m}\), but only to estimate a small number of
		quantities derived from it.
		
		\subsubsection{The harmonic correction problem and quantities of interest}
		
		Given the vortex configuration
		\[
		\bm a^m=(\bm a_1^m,\dots,\bm a_M^m)
		\]
		at time \(t_m\), the corresponding harmonic correction \(h_{\bm a^m}\) is
		determined by the boundary-value problem
		\begin{equation}\label{eq:ha-Dirichlet-sec3}
			\left\{
			\begin{aligned}
				\Delta h_{\bm a^m} &= 0
				\qquad &&\text{in }\Omega,\\
				h_{\bm a^m} &= \phi_g-\Theta_{\bm a^m}
				\qquad &&\text{on }\partial\Omega,
			\end{aligned}
			\right.
		\end{equation}
		where \(\Theta_{\bm a^m}\) is the singular phase associated with the current
		vortex locations and \(\phi_g\) is a boundary phase lifting of \(g\).
		
		The quantities extracted from \(h_{\bm a^m}\) serve two distinct purposes in the
		hybrid framework. First, in the coupled reduced model
		\eqref{eq:M2-model-sec3}, the classical update requires only the local feedback
		values
		\begin{equation}\label{eq:feedback-observables}
			\mathcal O_j^{m,\mathrm{fb}}
			:=
			\nabla h_{\bm a^m}(\bm a_j^m),
			\qquad j=1,\dots,M.
		\end{equation}
		Second, at selected time steps---in particular at the final time---one may wish
		to recover physical observables associated with the outer approximation
		\[
		u_{\bm a^m}(\bm x)
		=
		\exp\!\bigl(i(\Theta_{\bm a^m}(\bm x)+h_{\bm a^m}(\bm x))\bigr)
		\]
		and the corresponding effective velocity field
		\[
		\bm v_{\bm a^m}(\bm x)
		=
		\nabla\Theta_{\bm a^m}(\bm x)+\nabla h_{\bm a^m}(\bm x).
		\]
		For this purpose, the relevant quantities include point values
		\(h_{\bm a^m}(\bm x_\ell)\) and gradients
		\(\nabla h_{\bm a^m}(\bm x_\ell)\) at a small number of probe locations
		\(\bm x_\ell\), \(\ell=1,\dots,L_{\mathrm{obs}}\). Once these are available, the corresponding
		values of \(u_{\bm a^m}(\bm x_\ell)\) and \(\bm v_{\bm a^m}(\bm x_\ell)\) are
		recovered by straightforward classical post-processing. In both uses, the key
		observation is that one only needs a small number of local values or gradients,
		rather than the full field \(h_{\bm a^m}\).
		
		To make this small-output structure accessible to quantum linear-algebraic
		subroutines, we next turn to the spatial discretization of the harmonic
		correction problem.

		\subsubsection{Spatial discretization and discrete observables}
		
		After spatial discretization, \eqref{eq:ha-Dirichlet-sec3} gives rise to a
		linear system
		\begin{equation}\label{eq:discrete-harmonic-system}
			K\,\bm h_a^m=\bm b_a^m,
		\end{equation}
		where \(K\in\mathbb R^{N_h\times N_h}\) is the discrete Laplace operator,
		\(\bm h_a^m\in\mathbb R^{N_h}\) is the coefficient vector representing
		\(h_{\bm a^m}\), and \(\bm b_a^m\in\mathbb R^{N_h}\) is induced by the boundary
		data \(\phi_g-\Theta_{\bm a^m}\). Here \(N_h\) denotes the number of spatial
		degrees of freedom.
		
		A crucial feature is that the matrix \(K\) is independent of the vortex
		configuration, whereas all dependence on \(\bm a^m\) enters through the
		right-hand side \(\bm b_a^m\). Thus, once the spatial discretization is fixed,
		the large-scale linear operator remains unchanged throughout the time evolution,
		and only the input vector varies from one time step to the next.
		
		The observables described above become linear functionals of \(\bm h_a^m\). More
		precisely, for each observable of interest, there exists a vector
		\(\bm c_\ell\in\mathbb R^{N_h}\) such that
		\begin{equation}\label{eq:observable-functional}
			\mathcal O_\ell(\bm h_a^m)=\langle \bm c_\ell,\bm h_a^m\rangle,
			\qquad \ell=1,\dots,L_{\mathrm{obs}}.
		\end{equation}
		Collecting them together, we write
		\begin{equation}\label{eq:observable-matrix}
			\mathcal O(\bm h_a^m)=C\,\bm h_a^m,
		\end{equation}
		where \(C\in\mathbb R^{L_{\mathrm{obs}}\times N_h}\) is an observation matrix. In particular,
		the feedback observables \eqref{eq:feedback-observables} and the final probe
		values used to reconstruct \(u_{\bm a^m}\) and \(\bm v_{\bm a^m}\) are all of
		this form.
		
		With this discrete formulation in hand, we can now describe the quantum solver
		used to estimate these observables efficiently.
		
		\subsubsection{Quantum preconditioned solver}
		
		We now apply the Schr\"odingerization framework reviewed in
		Section~\ref{sec:schrodingerization} to the harmonic correction problem.
		Rather than solving \eqref{eq:discrete-harmonic-system} directly, we introduce
		the BPX preconditioner
		\[
		B=SS^\top
		\]
		for the Poisson operator \cite{BPX1990}, and rewrite the system in symmetrically
		preconditioned form as
		\begin{equation}\label{eq:preconditioned-harmonic-system}
			K_S\,\bm z_a^m=\bm b_{S,a}^m,
			\qquad
			K_S:=S^\top K S,
			\qquad
			\bm b_{S,a}^m:=S^\top \bm b_a^m,
		\end{equation}
		so that
		\begin{equation}\label{eq:recover-h-from-z}
			\bm h_a^m=S\,\bm z_a^m.
		\end{equation}
		This brings the linear correction problem exactly into the form required by the
		Schr\"odingerization-based quantum algorithm.
		
		A key point is that the quantum routine acts on the symmetrically preconditioned
		operator \(K_S\), rather than on the original matrix \(K\). As proved in
		\cite{JLMY25}, the matrix \(K_S\) admits a structure-aware exact block-encoding
		with normalization
		\begin{equation}\label{eq:encoding para Ks}
			O\!\bigl(d^2 \log(N_h)\bigr),
		\end{equation}
		where $d$ is the dimension.
		
		As explained in Section~\ref{sec:schrodingerization}, the system
		\eqref{eq:preconditioned-harmonic-system} is first embedded into an augmented
		ODE, then transformed by the warped phase representation into a Schr\"odinger-type
		equation in one higher dimension, and finally implemented by Hamiltonian
		simulation after truncation and spectral discretization in the auxiliary
		variable. Since our goal is not to reconstruct the full vector \(\bm h_a^m\),
		but only to estimate a small number of observables, we combine this evolution
		with the linear-measurement strategy reviewed earlier. Indeed, for any observable
		of the form
		\[
		\mathcal O_\ell(\bm h_a^m)=\langle \bm c_\ell,\bm h_a^m\rangle,
		\]
		one has
		\begin{equation}\label{eq:preconditioned-observable}
			\mathcal O_\ell(\bm h_a^m)
			=
			\langle \bm c_\ell,S\bm z_a^m\rangle
			=
			\langle S^\top \bm c_\ell,\bm z_a^m\rangle.
		\end{equation}
		Hence the desired quantity is reduced to a linear observable of the
		preconditioned unknown \(\bm z_a^m\). In particular, the feedback terms needed
		in the coupled vortex update, as well as the probe values used for the final
		reconstruction of \(u_{\bm a^m}\) and \(\bm v_{\bm a^m}\), all fit into this
		same measurement framework.
		
		Accordingly, the quantum part of the algorithm consists of four steps: preparing
		the input state associated with \(\bm b_{S,a}^m\), simulating the
		Schr\"odingerized Hamiltonian generated by \(K_S\), recovering the steady-state
		component corresponding to \(\bm z_a^m\), and finally estimating the target
		observable in \eqref{eq:preconditioned-observable} by measuring against the
		state prepared from \(S^\top\bm c_\ell\). In this way, quantum resources are
		used only for the high-dimensional linear harmonic correction problem, while the
		low-dimensional vortex dynamics remains classical.

	\subsection{Hybrid algorithmic workflow}
	\label{subsec:hybrid-algorithm}
	Two operating modes arise naturally. In the simplified model \emph{M1}, the vortex
dynamics is decoupled from the harmonic correction, so the vortex trajectory can
be computed entirely classically, and the quantum routine is invoked only at
selected output times. In the coupled model \emph{M2}, the vortex update depends on
the feedback values \(\nabla h_{\bm a^m}(\bm a_j^m)\), so the classical and
quantum parts must interact at each time step. In the numerical experiments
below, we will compare \emph{M1} and \emph{M2} to quantify the effect of the
boundary-induced harmonic correction on the reduced dynamics and outer
reconstruction. These two modes are summarized in
Table~\ref{tab:hybrid-modes}.

		\begin{table}[htbp]
		\centering
\begin{tabular}{p{0.18\textwidth} p{0.32\textwidth} p{0.32\textwidth}}
	\hline
	& \textbf{M1 (decoupled mode)} & \textbf{M2 (coupled mode)}  \\
	\hline
	Role of \(h_{\bm a}\) &
	Reconstruction only &
	Feedback + reconstruction \\[0.3em]
	Vortex update &
	Classical only &
	Classical with quantum feedback \\[0.3em]
	Quantum call &
	Only at selected output times &
	At each time step \\[0.3em]
	When needed &
	When boundary feedback is negligible &
	When boundary feedback is non-negligible \\
	\hline
\end{tabular}
			\caption{Two operating modes of the hybrid framework.}
			\label{tab:hybrid-modes}
		\end{table}

	For clarity, we summarize the two workflows in pseudocode form shown in Algorithm \ref{alg:hybrid-m1} and \ref{alg:hybrid-m2}, respectively.
	
	\begin{algorithm}[htbp]
		\caption{Hybrid solver in the decoupled mode \emph{M1}}
		\label{alg:hybrid-m1}
		\begin{algorithmic}[1]
			\State \textbf{Input:} \(\Omega\), \(\phi_g\), \(\bm a^0\), \(T\), \(\Delta t\), discretization data, observation vectors \(\{\bm c_\ell\}\)
            \State Set up the discrete operator \(K\), specify the BPX factor \(S\), and prepare quantum access to \(K_S=S^\top K S\)
			\For{$m=0,1,\dots,N_T-1$}
			\State Update \(\bm a^{m+1}\) by the classical free-space rule \eqref{eq:classical-vortex-update}
			\EndFor
			\State Choose an output time \(t_m\) (for example \(t_{N_T}=T\))
			\State Build \(\Theta_{\bm a^m}\), the boundary data \(\phi_g-\Theta_{\bm a^m}\), and the right-hand side \(\bm b_a^m\)
			\State Form the preconditioned input \(\bm b_{S,a}^m=S^\top \bm b_a^m\)
			\State Quantum routine: estimate the desired observables
			\[
			\mathcal O_\ell(\bm h_a^m)=\langle S^\top \bm c_\ell,\bm z_a^m\rangle,
			\qquad
			K_S\bm z_a^m=\bm b_{S,a}^m
			\]
			\State \textbf{Output:} trajectory \(\{\bm a^m\}_{m=0}^{N_T}\) and selected observables of \(u_{\bm a^m}\), \(\bm v_{\bm a^m}\)
		\end{algorithmic}
	\end{algorithm}
	
	\begin{algorithm}[htbp]
		\caption{Hybrid solver in the coupled mode \emph{M2}}
		\label{alg:hybrid-m2}
		\begin{algorithmic}[1]
			\State \textbf{Input:} \(\Omega\), \(\phi_g\), \(\bm a^0\), \(T\), \(\Delta t\), discretization data, feedback/observation vectors \(\{\bm c_\ell\}\)
            \State Set up the discrete operator \(K\), specify the BPX factor \(S\), and prepare quantum access to \(K_S=S^\top K S\)
			\For{$m=0,1,\dots,N_T-1$}
			\State Build \(\Theta_{\bm a^m}\), the boundary data \(\phi_g-\Theta_{\bm a^m}\), and the right-hand side \(\bm b_a^m\)
			\State Form the preconditioned input \(\bm b_{S,a}^m=S^\top \bm b_a^m\)
			\State Quantum routine: estimate the feedback values
			\[
			\nabla h_{\bm a^m}(\bm a_j^m),
			\qquad j=1,\dots,M,
			\]
			via observables of the form
			\[
			\mathcal O_\ell(\bm h_a^m)=\langle S^\top \bm c_\ell,\bm z_a^m\rangle,
			\qquad
			K_S\bm z_a^m=\bm b_{S,a}^m
			\]
			\State Update \(\bm a^{m+1}\) by the coupled rule \eqref{eq:classical-vortex-update-coupled}
			\EndFor
			\State \textbf{Output:} coupled trajectory \(\{\bm a^m\}_{m=0}^{N_T}\) and selected observables of \(u_{\bm a^m}\), \(\bm v_{\bm a^m}\)
		\end{algorithmic}
	\end{algorithm}
	
	In both modes, the quantum subroutine is used only to estimate a small number of
	observables associated with the harmonic correction, rather than to reconstruct
	the full field \(h_{\bm a^m}\) on the whole mesh. Thus the amount of
	information exchanged between the classical and quantum components remains
	small. In particular, even in the coupled mode \emph{M2}, the online communication is
	limited to the feedback quantities needed for the vortex update, which preserves
	the small-output character of the overall framework.

		\section{Ginzburg–Landau vortex dynamics and filamentary extensions}
		\label{sec:GL and 3D extensions}
        
		In this section we extend the idea of Section~\ref{sec:reduction}
		to {\it dissipative} equation, the nonlinear heat or Ginzburg-Landau equation, which gives rise to  Ginzburg--Landau vortex dynamics and to filamentary vortex structures in
		three dimensions. The main point is that, although the reduced geometric objects
		and their motion laws become different, the same basic mechanism persists: the
		nonlinear singular structure is described by a low-dimensional vortex dynamics,
		whereas the remaining outer field is governed by a linear equation. We first
		discuss the  Ginzburg--Landau model in two dimensions, and then turn
		to the three-dimensional superconductivity setting, where point vortices are
		replaced by vortex filaments.

		\subsection{Vortex motion in Ginzburg--Landau models}
		\label{subsec:vortex in Ginzburg--Landau}

		Throughout this subsection, we adopt the same geometric setting, the same-sign
		unit-vortex regime, and the corresponding vortex-regime assumptions as in
		Section~\ref{sec:reduction}. In particular, the reduced variables are again the
		vortex locations
		\[
		\bm a(t)=(\bm a_1(t),\dots,\bm a_M(t))\subset\Omega.
		\]
		
		Within the unified framework for complex scalar field equations
		\cite{LinXin1998}, we consider the logarithmically rescaled nonlinear heat
		equation
		\begin{equation}\label{eq:GL_heat}
			\frac{1}{\log(1/\varepsilon)}\,\partial_t u^\varepsilon
			=
			\Delta u^\varepsilon
			+\frac{1}{\varepsilon^2}(1-|u^\varepsilon|^2)u^\varepsilon,
			\qquad (\bm x,t)\in\Omega\times(0,T).
		\end{equation}
		As in the nonlinear Schr\"odinger case, in the vortex regime one has
		\(|u^\varepsilon|\approx 1\) away from the vortex cores. Hence the leading
		outer behavior is again described by a singular vortex phase together with a
		regular correction, namely
		\begin{equation}\label{eq:GL_outer_ansatz}
			u^\varepsilon(\bm x,t)\approx
			e^{i(\Theta_{\bm a(t)}(\bm x)+h_{\bm a(t)}(\bm x))},
		\end{equation}
		where
		$
		\Theta_{\bm a}(\bm x)=\sum_{j=1}^M \arg(\bm x-\bm a_j).
		$
		Thus, at the level of outer structure, the same nonlinear-to-linear
		decomposition persists: the singular part is explicitly determined by the vortex
		locations, whereas the regular part is encoded by \(h_{\bm a}\).
		
		The essential difference from the nonlinear Schr\"odinger case lies in the
		effective vortex motion law. In the nonlinear heat / Ginzburg--Landau setting,
		the vortex locations evolve according to the gradient-flow law
		\begin{equation}\label{eq:GL_gradient_law}
			\dot{\bm a}_j(t)=-\nabla_{\bm a_j}W(\bm a(t)),
			\qquad j=1,\dots,M,
		\end{equation}
		where \(W(\bm a)\) is the renormalized energy introduced in
		\eqref{eq:renormalized-energy}. Thus the reduced vortex degrees of freedom
		remain the same, but the effective dynamics changes from the Hamiltonian law
		\(J\nabla_{\bm a_j}W\) in the NLS case to the dissipative gradient flow
		\(-\nabla_{\bm a_j}W\) in the nonlinear heat / Ginzburg--Landau case.
		
		For bounded domains, the boundary contribution is again encoded by the same
		regular correction \(h_{\bm a}\) as in
		Section~\ref{sec:harmonic correction}. Accordingly, keeping the same
		normalization convention as in \eqref{eq:M2-model}, the gradient-flow law
		\eqref{eq:GL_gradient_law} takes the explicit form
		\begin{equation}\label{eq:GL_bounded_motion}
			\dot{\bm a}_j(t)
			=
			2\sum_{k\ne j}\frac{\bm a_j-\bm a_k}{|\bm a_j-\bm a_k|^2}
			-
			2\nabla^\perp h_{\bm a}(\bm a_j),
			\qquad j=1,\dots,M,
		\end{equation}
		where \(\nabla^\perp f:=(-\partial_{x_2}f,\partial_{x_1}f)\). In particular,
		the vortex dynamics consists of an explicit pairwise interaction together with a
		boundary-induced drift encoded by the regular correction.
		
		When the boundary effect is negligible over the time interval of interest, one
		may further adopt the free-space approximation
		\begin{equation}\label{eq:GL_plane_motion}
			\dot{\bm a}_j(t)
			=
			2\sum_{k\ne j}\frac{\bm a_j-\bm a_k}{|\bm a_j-\bm a_k|^2},
			\qquad j=1,\dots,M.
		\end{equation}
		This free-space law should be understood as an auxiliary approximation, rather
		than the rigorous bounded-domain motion law.

		For the present hybrid framework, the low-dimensional system
		\eqref{eq:GL_gradient_law} is still evolved classically, while the
		quantum-relevant part remains the regular correction associated with each frozen
		vortex configuration. In the reduced description adopted here, \(h_{\bm a}\) is
		therefore treated as the same harmonic correction as in
		Section~\ref{sec:harmonic correction}; more precisely, for Dirichlet boundary
		data \(g=e^{i\phi_g}\), it is recovered from the elliptic problem
		\begin{equation}\label{eq:GL_ha_elliptic}
			\Delta h_{\bm a}=0
			\qquad \text{in }\Omega,\qquad
			h_{\bm a}=\phi_g-\Theta_{\bm a}
			\qquad \text{on }\partial\Omega,
		\end{equation}
		modulo the natural \(2\pi\)-ambiguity of the phase. Thus, at each time step,
		the vortex locations are advanced by the gradient flow
		\eqref{eq:GL_gradient_law}, while the regular field is obtained by solving the
		linear elliptic problem \eqref{eq:GL_ha_elliptic}. We emphasize that this is a
		quasi-static reduced description on the logarithmically accelerated
		vortex-motion scale; on the usual diffusive scale, one may instead formally
		derive \(\log(1/\varepsilon)^{-1}\partial_t \phi=\Delta \phi\) for the outer
		phase.

		We also note that the Ginzburg--Landau perspective is not restricted to planar
		point vortices. In particular, \cite{E1994} further treats superconductivity
		models and derives leading-order laws for vortex lines in three dimensions.
		This provides the natural bridge to the filamentary setting considered next.

		\subsection{From point vortices to vortex filaments in three-dimensional superconductivity}
		To pass from two-dimensional point vortices to three-dimensional filamentary
		structures while retaining an explicit linear outer problem, we consider the
		three-dimensional time-dependent Ginzburg--Landau model for superconductivity
		studied in \cite{E1994}. The primitive unknowns are the complex order parameter
		\(\psi\), the magnetic vector potential \(A\), and the electric potential
		\(\Phi\), with induced magnetic field
		\[
		H=\nabla\times A.
		\]
		Since the system is gauge-invariant, we work in the London gauge
		\begin{equation}\label{eq:London-gauge}
			\nabla\cdot A=0.
		\end{equation}
		In nondimensional variables,  the model takes the form
		\begin{equation}\label{eq:TDGL-3D}
			\left\{
			\begin{aligned}
				\gamma\bigl(\partial_t \psi+i\Phi\psi\bigr)
				&=
				-\bigl(-i\varepsilon\nabla-A\bigr)^2\psi
				+(1-|\psi|^2)\psi,\\
				\sigma\bigl(\partial_t A+\nabla\Phi\bigr)
				&=
				-\nabla\times(\nabla\times A)
				+\frac{1}{2i}\bigl(\overline{\psi}\,\varepsilon\nabla\psi
				-\psi\,\varepsilon\nabla\overline{\psi}\bigr)
				-|\psi|^2A,\\
				\nabla\cdot A&=0.
			\end{aligned}
			\right.
		\end{equation}
		Here \(\varepsilon\) characterizes the vortex-core scale, \(\gamma\) is a
		relaxation parameter, and \(\sigma\) is the normal-state conductivity.
		Compared with the scalar no-magnetic-field model in
		Section~\ref{subsec:vortex in Ginzburg--Landau}, \eqref{eq:TDGL-3D} is
		substantially more involved. Nevertheless, in the vortex regime relevant here it
		still admits an asymptotically reduced description, consisting of a geometric
		evolution law for the vortex filaments together with a linear outer problem for
		the magnetic field. This makes the model a natural three-dimensional extension
		of the reduction principle developed above.
		
		In three dimensions, the singular set is an evolving vortex filament
		\[
		\Gamma(t)=\{\bm X(s,t):s\in I\}\subset\Omega\subset\mathbb R^3,
		\]
		or, more generally, a finite union of such curves. Here \(\bm X(s,t)\) is a
		parametrization of the filament and \(s\) is an arclength-type parameter. This
		filament plays the role that the point-vortex locations \(\bm a_j(t)\) played in
		two dimensions: it is the reduced geometric object carrying the singular
		topological structure of the solution.
		
		The matched asymptotic reduction of \eqref{eq:TDGL-3D} yields two leading-order
		objects: a geometric evolution law for the filament itself, and a linear outer
		problem for the magnetic field. At leading order, \cite{E1994} showed that vortex
		filaments move by curvature flow in the normal direction,
		\begin{equation}\label{eq:filament_curvature_flow}
			\bm V=\kappa \bm n,
		\end{equation}
		where \(\bm V\) is the filament velocity, \(\kappa\) is the curvature, and
		\(\bm n\) is the unit normal.
		
		Once the filament geometry is prescribed, the remaining high-dimensional
		component is the outer magnetic field. This linear part arises from the outer
		expansion of the full superconductivity system. In the asymptotic regime
		considered in \cite{E1994}, the magnetic flux carried by each vortex is of order
		\(O(\varepsilon)\). Consequently, in the outer region the electromagnetic
		variables are expanded starting at order \(\varepsilon\), while the order
		parameter and the electric potential admit the expansions
		\[
		\psi=\psi_0+\varepsilon\psi_1+\varepsilon^2\psi_2+\cdots,\qquad
		A=\varepsilon A_1+\varepsilon^2A_2+\cdots,\qquad
		\Phi=\Phi_0+\varepsilon\Phi_1+\varepsilon^2\Phi_2+\cdots.
		\]
		Accordingly,
		\[
		H=\nabla\times A=\varepsilon H_1+\varepsilon^2H_2+\cdots,
		\qquad H_1:=\nabla\times A_1,
		\]
		so that \(H_1\) is the first nontrivial outer magnetic-field component.
		
		Formally, substituting these expansions into the gauge-fixed system
		\eqref{eq:TDGL-3D} and collecting terms of the same order, the \(O(1)\) part of
		the \(\psi\)-equation yields
		\[
		\gamma\bigl(\partial_t\psi_0+i\Phi_0\psi_0\bigr)
		=
		(1-|\psi_0|^2)\psi_0.
		\]
		In the outer region the leading state lies on the superconducting manifold, so
		that \(|\psi_0|=1\), and hence
		\[
		\partial_t\psi_0+i\Phi_0\psi_0=0.
		\]
		In particular, as in the outer analysis of \cite{E1994}, one may write
		\[
		\psi_0=e^{i\theta_0},\qquad \Phi_0=0.
		\]
		
		Next, collecting the \(O(\varepsilon)\) terms in the \(A\)-equation gives
		\[
		\sigma(\partial_t A_1+\nabla\Phi_1)
		=
		-\nabla\times(\nabla\times A_1)
		+\frac{1}{2i}\bigl(\overline{\psi_0}\nabla\psi_0
		-\psi_0\nabla\overline{\psi_0}\bigr)
		-|\psi_0|^2A_1.
		\]
		Using \(\psi_0=e^{i\theta_0}\), one has
		\[
		\frac{1}{2i}\bigl(\overline{\psi_0}\nabla\psi_0
		-\psi_0\nabla\overline{\psi_0}\bigr)=\nabla\theta_0,
		\qquad
		|\psi_0|^2=1,
		\]
		and therefore
		\[
		\sigma(\partial_t A_1+\nabla\Phi_1)
		=
		-\nabla\times(\nabla\times A_1)+\nabla\theta_0-A_1.
		\]
		
		Now impose the London gauge at leading order,
		\[
		\nabla\cdot A_1=0.
		\]
		Taking the curl of the previous equation and using
		\[
		\nabla\times\nabla\Phi_1=0,\qquad
		\nabla\times\nabla\theta_0=0
		\]
		away from the vortex core, one obtains
		\[
		\sigma\,\partial_t H_1
		=
		\Delta H_1-H_1,
		\qquad H_1=\nabla\times A_1.
		\]
		Hence, in the outer leading-order quasi-static balance, the magnetic field
		satisfies the homogeneous London equation on the punctured domain,
		\[
		-\Delta H_1+H_1=0.
		\]
		
		In three dimensions, however, the inner vortex core is seen from the outer scale
		as a filament \(\Gamma(t)\). The matching condition at the core is therefore
		encoded by a source term supported on \(\Gamma(t)\), which yields the compact
		London equation
		\begin{equation}\label{eq:London-filament}
			-\Delta H_1 + H_1
			=
			2\pi\int_{\Gamma(t)}
			\delta(\bm x-\bm X(s,t))\,\bm \ell(s,t)\,ds,
		\end{equation}
		where \(\bm \ell(s,t)\) is the unit tangent vector to the filament. Equivalently,
		\begin{equation}\label{eq:London-filament-green}
			H_1(\bm x,t)
			=
			2\pi\int_{\Gamma(t)}
			G(\bm x-\bm X(s,t))\,\bm \ell(s,t)\,ds,
		\end{equation}
		with \(G\) the Green function of \((-\Delta+I)^{-1}\) in \(\mathbb R^3\).
		
		We stress that the quantum target in the present three-dimensional setting is not
		the full gauge-dependent system \((\psi,A,\Phi)\) itself. Rather, after
		asymptotic reduction, the quantities of primary interest are the filament
		geometry and the corresponding leading outer magnetic field. The order parameter
		\(\psi\) still contains the full nonlinear vortex-core structure, while the
		electric potential \(\Phi\) remains part of the gauge-fixed coupled system and is
		not singled out as an explicit linear observable in the reduction of
		\cite{E1994}. By contrast, the leading outer magnetic field \(H_1\) is governed
		by a closed linear equation once the filament geometry is prescribed. For this
		reason, \(H_1\) is the natural high-dimensional component to be handled by the
		quantum routine, whereas the nonlinear geometric evolution of the filament is
		advanced classically.
		
		Hence the three-dimensional superconductivity setting retains the same basic
		splitting as before: a reduced classical evolution for the filament geometry and
		a linear outer field equation for the quantum routine.

		\section{The hybrid quantum-classical algorithms for the Ginzburg-Landau models}
	\label{sec:computational perspectives}
		The computational bottleneck of the reduced algorithm lies in the linear
		correction stage rather than in the low-dimensional vortex update. We therefore
		focus on the mesh-dependent linear solver together with the evaluation of a small
		number of observables.
		
		For the two-dimensional harmonic correction problem, at each time step \(t_m\),
		once the vortex configuration \(\bm a^m\) has been advanced classically, the
		remaining task is the solution of
		\begin{equation}\label{eq:Kx=b}
			K\,\bm h_a^m=\bm b_a^m,
		\end{equation}
		together with the evaluation of
		\[
		\mathcal O_\ell^m=\langle \bm c_\ell,\bm h_a^m\rangle,
		\qquad \ell=1,\dots,L_{\mathrm{obs}},
		\]
		where typically \(L_{\mathrm{obs}}\ll N_h\), and \(N_h\) denotes the number of
		spatial degrees of freedom. On the quantum side, we use the BPX factorization
		\(B=SS^\top\) and the symmetrically preconditioned system
		\[
		K_S:=S^\top K S,
		\qquad
		\bm b_{S,a}^m:=S^\top \bm b_a^m,
		\]
		where the original unknown is recovered from the preconditioned one through
		\[
		\bm h_a^m=S\bm z_a^m.
		\]
		Then
		\[
		\mathcal O_\ell^m
		=
		\langle \bm c_\ell,\bm h_a^m\rangle
		=
		\langle S^\top \bm c_\ell,\bm z_a^m\rangle,
		\]
		so the desired output reduces to a linear observable of the preconditioned
		unknown.
		
		We next consider the three-dimensional time-dependent Ginzburg--Landau model for
		superconductivity in the London regime. At the linear level, the outer magnetic
		field satisfies
		\begin{equation}\label{eq:London-filament}
			-\Delta H_1 + H_1
			=
			2\pi\int_{\Gamma(t)}
			\delta(\bm x-\bm X(s,t))\,\bm \ell(s,t)\,ds.
		\end{equation}
		After spatial discretization on a three-dimensional dyadic grid with mesh size
		\(h\), this gives rise to a linear system of the form
		\begin{equation}\label{eq:discrete-London}
			(K+h^d I)\bm H_1=\bm f_\Gamma,
		\end{equation}
		where \(K\) is the same discrete Laplace operator as in the Poisson case, and
		\(d=3\) in the present setting. We again employ the BPX preconditioner
		\(B=SS^\top\) and consider the symmetrically preconditioned matrix
		\begin{equation}\label{eq:shifted-preconditioned-matrix}
			\widetilde K_S:=S^\top(K+h^d I)S
			=
			K_S+h^d S^\top S,
			\qquad
			K_S:=S^\top K S.
		\end{equation}
		
		\begin{lemma}[Uniform spectral bounds for the shifted preconditioned operator]
			\label{lem:shifted-bpx}
			Let
			\[
			K_S:=S^\top K S,
			\qquad
			\widetilde K_S:=S^\top(K+h^d I)S
			=
			K_S+h^d S^\top S.
			\]
			Assume that the preconditioned matrix \(K_S\) has mesh-uniform spectral bounds,
			namely,
			\[
			\lambda_{\min}(K_S)\ge c_0,
			\qquad
			\lambda_{\max}(K_S)\le C_0,
			\]
			for some constants \(c_0,C_0>0\) independent of \(h\). Then there exists a
			constant \(C_*>0\), also independent of \(h\), such that
			\[
			\lambda_{\min}(\widetilde K_S)\ge c_0,
			\qquad
			\lambda_{\max}(\widetilde K_S)\le C_*.
			\]
			In particular, \(\widetilde K_S\) also has mesh-uniform spectral bounds.
		\end{lemma}
		
		\begin{proof}
			For any \(\bm z\), let \(\bm v=S\bm z\). By the discrete Poincar\'e inequality
			on the Dirichlet grid, there exists a constant \(C_P>0\), independent of \(h\),
			such that
			\[
			h^d\|\bm v\|_2^2 \le C_P\,\bm v^\top K\bm v.
			\]
			Hence
			\[
			h^d\bm z^\top S^\top S\bm z
			=
			h^d\|S\bm z\|_2^2
			\le
			C_P\,\bm z^\top K_S\bm z.
			\]
			Therefore
			\[
			\bm z^\top \widetilde K_S \bm z
			=
			\bm z^\top K_S\bm z+h^d\bm z^\top S^\top S\bm z
			\le
			(1+C_P)\,\bm z^\top K_S\bm z.
			\]
			On the other hand, since \(h^dS^\top S\) is positive semidefinite,
			\[
			\bm z^\top \widetilde K_S \bm z \ge \bm z^\top K_S\bm z.
			\]
			Using the spectral bounds of \(K_S\),
			\[
			c_0\|\bm z\|_2^2
			\le
			\bm z^\top K_S\bm z
			\le
			C_0\|\bm z\|_2^2,
			\]
			we obtain
			\[
			c_0\|\bm z\|_2^2
			\le
			\bm z^\top \widetilde K_S \bm z
			\le
			(1+C_P)C_0\|\bm z\|_2^2.
			\]
			Thus
			\[
			\lambda_{\min}(\widetilde K_S)\ge c_0,
			\qquad
			\lambda_{\max}(\widetilde K_S)\le (1+C_P)C_0.
			\]
			Taking \(C_*:=(1+C_P)C_0\) completes the proof.
		\end{proof}
		
		The two-dimensional harmonic correction problem and the three-dimensional
		London-type system therefore fall into the same BPX--Schr\"odingerization
		framework. In both cases, the relevant output is a small number of linear
		observables of the preconditioned unknown. This leads to the following unified
		complexity comparison.
		
		\begin{theorem}[Complexity comparison for the linear correction stage]
			\label{thm:computational-perspective}
			Consider one time step of the reduced algorithm, and let the corresponding
			preconditioned linear system be written in the form
			\[
			\mathcal K_S \bm z=\bm g,
			\]
			where \( \mathcal K_S=K_S\) in the two-dimensional harmonic correction case,
			and \( \mathcal K_S=\widetilde K_S\) in the three-dimensional London-type case.
			Let the target observable be
			\[
			\mathcal O_\ell=\langle \bm c_{S,\ell},\bm z\rangle,
			\]
			where $\bm c_{S,\ell} = S^{\top} \bm c_\ell$.
			Then the following hold.
			
			\medskip
			\noindent
			\textup{(i) Classical cost.}
			A classical BPX-preconditioned iterative solver computes one right-hand side to
			accuracy \(\mathrm{tol}\) in
			\[
			\mathcal O\!\bigl(N_h\log(1/\mathrm{tol})\bigr)
			\]
			arithmetic operations, up to implementation-dependent polylogarithmic factors.
			
			\medskip
			\noindent
			\textup{(ii) Quantum cost.}
			Using the Schr\"odingerization-based quantum algorithm introduced above
			together with the BPX block-encoding of \( \mathcal K_S\), one can, for any
			target accuracy \(\mathrm{tol}\in(0,1)\), output an estimate
			\(\widetilde{\mathcal O}_\ell\) satisfying
			\[
			\bigl|\widetilde{\mathcal O}_\ell-\mathcal O_\ell\bigr|
			\le \mathrm{tol},
			\]
			with query complexity
			\[
			\mathcal O\!\bigl(
			\operatorname{poly}(d)\,\mathrm{tol}^{-1}\,
			\operatorname{polylog}(N_h)\,
			\operatorname{polylog}(1/\mathrm{tol})
			\bigr).
			\]
		\end{theorem}
		
		\begin{proof}
			The classical estimate follows from the sparsity of the discretized operator
			and the mesh-uniform spectral bounds of the BPX-preconditioned system. In
			particular, after BPX preconditioning the iterative error decays geometrically
			with a mesh-independent contraction factor, so the number of iterations needed
			to reach accuracy \(\mathrm{tol}\) is \(O(\log(1/\mathrm{tol}))\). Since each
			iteration costs \(O(N_h)\), the total classical complexity is
			\(\mathcal O\!\bigl(N_h\log(1/\mathrm{tol})\bigr)\). The quantum estimate
			follows directly from Theorem~\ref{thm:complexity}, together with the
			block-encoding bound \eqref{eq:encoding para Ks} in the two-dimensional case
			and Lemma~\ref{lem:shifted-bpx} in the three-dimensional shifted case.
		\end{proof}
		
		\begin{remark}
			Theorem~\ref{thm:computational-perspective} identifies the relevant regime for a
			possible quantum advantage: the linear correction stage with small-output
			readout. Classically, even with optimal preconditioning, one still pays at
			least linearly in the number of spatial degrees of freedom \(N_h\). By
			contrast, the Schr\"odingerization-based quantum algorithm yields a cost that
			depends on \(N_h\) only through polylogarithmic factors when one estimates a
			fixed number of observables.
			
			To compare the two costs in terms of the target accuracy \(\mathrm{tol}\), we
			consider a \(d\)-dimensional uniform grid with \(N_h=N_x^d\). Since our spatial
			discretization is first-order accurate, one has \(O(h)\) spatial error with
			\(h\sim N_x^{-1}\). Matching this with the target tolerance yields
			\[
			N_x\sim \mathrm{tol}^{-1},
			\qquad
			N_h\sim \mathrm{tol}^{-d}.
			\]
			Accordingly, the classical complexity becomes
			\[
			\mathcal O\!\bigl(\mathrm{tol}^{-d}\log(1/\mathrm{tol})\bigr),
			\]
			whereas the quantum complexity for estimating each observable becomes
			\[
			\mathcal O\!\bigl(\mathrm{tol}^{-1}\polylog(1/\mathrm{tol})\bigr).
			\]
			This does not imply an advantage for recovering the full field
			\(\bm h_a^m\); rather, the gain appears precisely when only a few physically
			relevant quantities are needed, such as local gradients, probe values, or
			other linear functionals of the outer expansion.
            In such cases our algorithms have {\it exponential} advantage over the classical algorithms in the dependence on \(N_h\), while
           its dependence on \(\mathrm{tol}^{-1}\) remains essentially linear up to
           polylogarithmic factors.
          
		\end{remark}
		
\section{Numerical verification of the nonlinear-to-linear decomposition}
\label{subsec:numerical-verification}

The quantum solution of the linear subproblems arising from elliptic or parabolic
equations has already been studied extensively in the literature. Therefore, the
purpose of the present numerical experiments is not to revalidate the linear
quantum solver itself, but rather to assess the PDE reduction developed in
Section~\ref{sec:reduction}. More precisely, we focus on two questions:
whether the reduced outer reconstruction approximates the full nonlinear
solution well away from the vortex cores, and whether the boundary-induced
\(h_{\bm a}\)-drift term is numerically relevant in the present regime.

We consider the square domain
$
\Omega=[0,1]^2\subset\mathbb R^2
$
with Dirichlet boundary condition
\[
u^\varepsilon(t,\bm x)=g(\bm x)=e^{i\phi(\bm x)},
\qquad \bm x\in\partial\Omega,
\]
where
\begin{equation}\label{eq:numerical-phi-choice-revised-final}
	\phi(x,y)
	=
	2\,\operatorname{atan2}\!\Bigl(y-\tfrac12,\;x-\tfrac12\Bigr)
	+
	0.3\,\sin(2\pi x)\sin(2\pi y).
\end{equation}
The first term gives boundary degree \(2\), while the second term is a smooth
perturbation that preserves the degree. The initial state contains two
degree-\(+1\) vortices located at
\begin{equation}\label{eq:numerical-vortex-initial-revised-final}
	\bm a_1(0)=(0.35,0.55),\qquad
	\bm a_2(0)=(0.70,0.40).
\end{equation}

We compare the two reduced vortex models introduced earlier: the free-space
approximation \emph{M1} in \eqref{eq:M1-model}, and the bounded-domain coupled model
M2 in \eqref{eq:M2-model}. Their
vortex trajectories are denoted by
\[
\bm a^{(1)}(t),\qquad \bm a^{(2)}(t).
\]
For a given vortex configuration \(\bm a(t)\), the harmonic correction
\(h_{\bm a(t)}\) is computed from the linear boundary-value problem in
Section~\ref{sec:harmonic correction}, and the corresponding reduced outer field
is reconstructed as
\[
u_{\bm a}(t,\bm x)
=
\exp\bigl(i(\Theta_{\bm a(t)}(\bm x)+h_{\bm a(t)}(\bm x))\bigr).
\]
Thus \emph{M1} uses \(\bm a^{(1)}(t)\) in the reconstruction, whereas \emph{M2} uses
\(\bm a^{(2)}(t)\).

The full nonlinear solution of \eqref{eq:NLS-2D} is computed with initial data
\begin{equation}\label{eq:numerical-initial-data-revised-final}
	u^\varepsilon(0,\bm x)
	=
	\Bigg(\prod_{j=1}^{2}
	f\Big(\tfrac{|\bm x-\bm a_j(0)|}{\varepsilon}\Big)\Bigg)
	\exp\bigl(i(\Theta_{\bm a(0)}(\bm x)+h_{\bm a(0)}(\bm x))\bigr),
\end{equation}
where \(f(\rho)\approx \tanh(\rho/\sqrt{2})\) is a standard radial core profile.
The harmonic problem is discretized by the standard five-point finite-difference
method and solved by a sparse linear solver, while the nonlinear equation is
evolved by a Strang splitting method on a sufficiently fine mesh with a small
time step. In all experiments below, the final time is fixed as
$
T=0.05.
$

Since no exact solution is available, we use the numerically computed full
solution \(u^\varepsilon\) as the reference solution. To avoid contamination
from the singular vortex cores, all errors are measured on a common masked
region. In the present experiments, the mask is centered at the reference
vortex locations and uses a fixed radius
$
r_{\mathrm{mask}}=0.1.
$
Accordingly, we define
\begin{equation}\label{eq:common-mask-region-final}
	\Omega_{\mathrm{mask}}(t)
	:=
	\Omega\setminus\bigcup_{j=1}^{2}
	B_{r_{\mathrm{mask}}}\bigl(\bm a^{(2)}_j(t)\bigr).
\end{equation}
This ensures that \emph{M1} and \emph{M2} are compared on the same outer region at each time,
while the mask size is kept identical for all tested values of \(\varepsilon\).
The meaning of the mask can be seen from the left panel of
Figure~\ref{fig:mask-and-phase}: the dashed circles indicate the disks
removed from the computational domain, and the error is evaluated only outside
these circles.

Let \(u_{\bm a^{(1)}}\) and \(u_{\bm a^{(2)}}\) be the reduced outer
reconstructions corresponding to \emph{M1} and \emph{M2}. We define
\begin{equation}\label{eq:numerical-boundary-effect-errors-revised-final}
	E_{M1}^\varepsilon(t)
	:=
	\|u^\varepsilon(t,\cdot)-u_{\bm a^{(1)}}(t,\cdot)\|_{L^2(\Omega_{\mathrm{mask}}(t))},
	\qquad
	E_{M2}^\varepsilon(t)
	:=
	\|u^\varepsilon(t,\cdot)-u_{\bm a^{(2)}}(t,\cdot)\|_{L^2(\Omega_{\mathrm{mask}}(t))}.
\end{equation}
Before computing the \(L^2\)-difference, we align the reconstructed field with
the reference solution by a global constant phase on
\(\Omega_{\mathrm{mask}}(t)\). In the numerical implementation, we report the
corresponding aligned relative \(L^2\)-errors.

Figure~\ref{fig:vortex-trajectories-final} shows the vortex trajectories
generated by \emph{M1} and \emph{M2} for \(\varepsilon=0.025\). Over the short time interval
considered here, the two reduced trajectories remain qualitatively close, but a
visible separation is already present near the final time. This indicates that
the boundary-induced \(h_{\bm a}\)-drift is not dominant at the trajectory level
for the present time horizon, yet it is not negligible either.
\begin{figure}[t]
	\centering
	\includegraphics[width=0.52\textwidth]{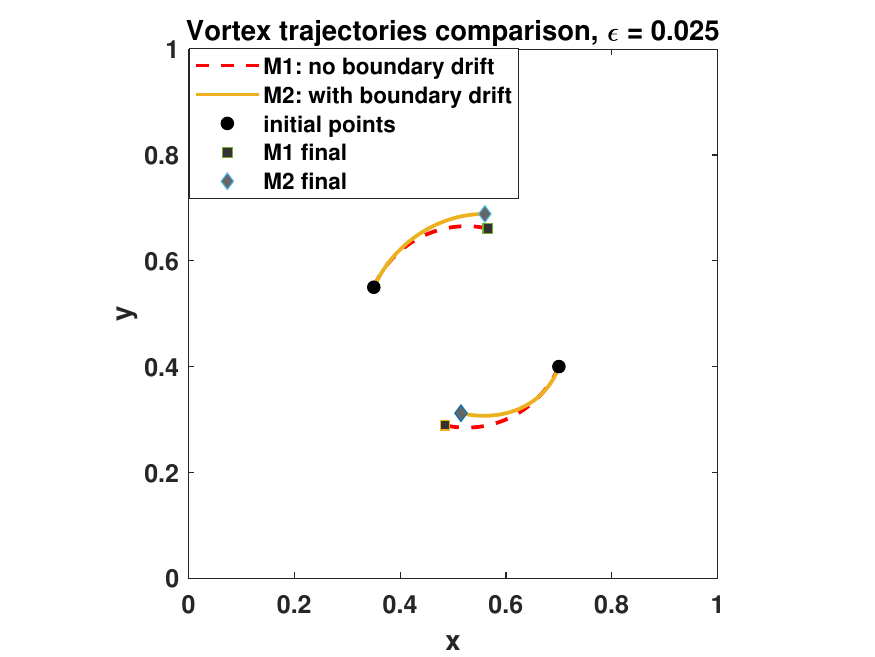}
	\caption{Vortex trajectories generated by \emph{M1} and \emph{M2} for
		\(\varepsilon=0.025\).}
	\label{fig:vortex-trajectories-final}
\end{figure}

To compare the phase structures, we use the phase mismatch
\begin{equation}\label{eq:phase-mismatch-revised-final}
	\delta_\phi(x,t)
	:=
	\Arg\!\bigl(
	u^\varepsilon(x,t)\,\overline{u_{\bm a}(x,t)}
	\bigr).
\end{equation}
Figure~\ref{fig:mask-and-phase} shows the common masked region and the phase
comparison at the final time. In the left panel, the background colors represent
the modulus \(|u^\varepsilon|\): bright colors correspond to values close to
\(1\), while dark blue regions indicate small values of \(|u^\varepsilon|\).
These dark blue regions correspond to the vortex cores and their nearby
neighborhoods, since \(|u^\varepsilon|\approx 0\) near a vortex center. The
dashed circles mark the excluded core disks used in the common mask. In the
right panel, the phase mismatch plots show that, away from these excluded core
regions, the reduced outer reconstruction captures the phase structure of the
full solution well. 

\begin{figure}[htbp]
	\centering
	\begin{minipage}[t]{0.48\textwidth}
			\centering
			\includegraphics[width=\textwidth]{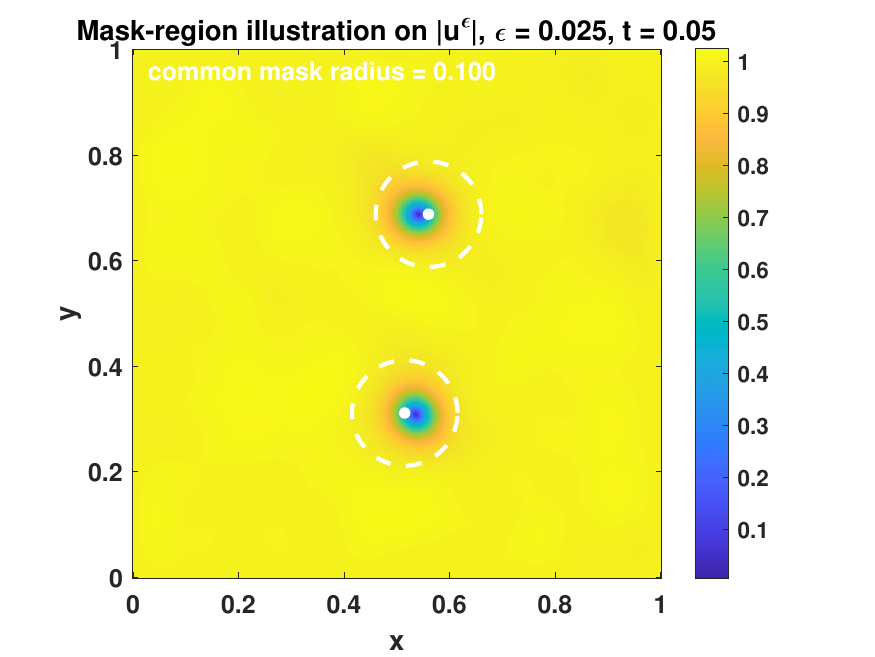}
		\end{minipage}\hfill
	\begin{minipage}[t]{0.51\textwidth}
			\centering
			\includegraphics[width=\textwidth]{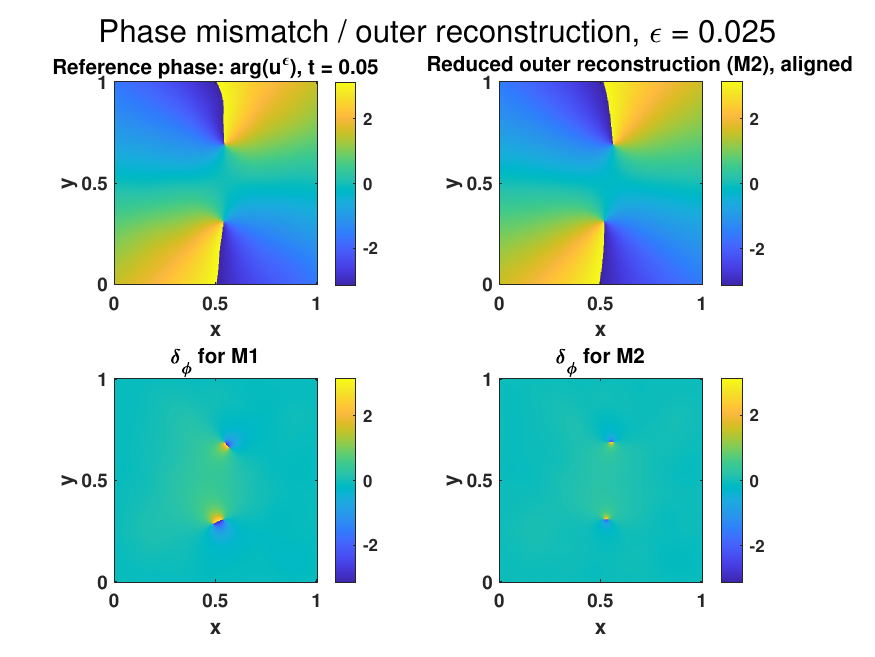}
		\end{minipage}
	\caption{Mask region and phase comparison.}
	\label{fig:mask-and-phase}
\end{figure}

Finally, Figure~\ref{fig:error-vs-epsilon-final} reports the final-time masked
errors for
\[
\varepsilon\in\{0.2,\,0.1,\,0.05,\,0.025\}.
\]
For both reduced models, the error decreases with a rate roughly of  $O(\varepsilon)$,
which is consistent with the validity of the vortex-based outer approximation in
the small-\(\varepsilon\) regime. At the same time, \emph{M2} consistently
outperforms \emph{M1}, and the gap becomes more visible for smaller \(\varepsilon\).
For instance, at \(\varepsilon=0.025\), the final masked error of \emph{M2} is
approximately one half of that of \emph{M1}. This shows that the boundary-induced
\(h_{\bm a}\)-drift term provides a quantitatively meaningful improvement in the
small-core regime.

We emphasize, however, that incorporating this correction does not change the
basic hybrid structure of the algorithm. Even for \emph{M2}, the classical part still
evolves only the low-dimensional vortex variables, while the quantum part is
used only for a small number of linear quantities associated with the harmonic
correction \(h_{\bm a}\), rather than for reconstructing the full field on the
entire mesh. Hence the amount of information exchanged between the classical and
quantum components remains small, so the potential quantum advantage of the
overall framework is preserved.

Overall, these experiments support the nonlinear-to-linear decomposition
developed in Section~\ref{sec:reduction}. They show that the reduced outer field
already provides a good approximation away from the vortex cores, and that the
linear harmonic correction encoded in \(h_{\bm a}\) yields a quantitatively
useful improvement, especially for smaller values of \(\varepsilon\).

\begin{figure}[t]
	\centering
	\includegraphics[width=0.60\textwidth]{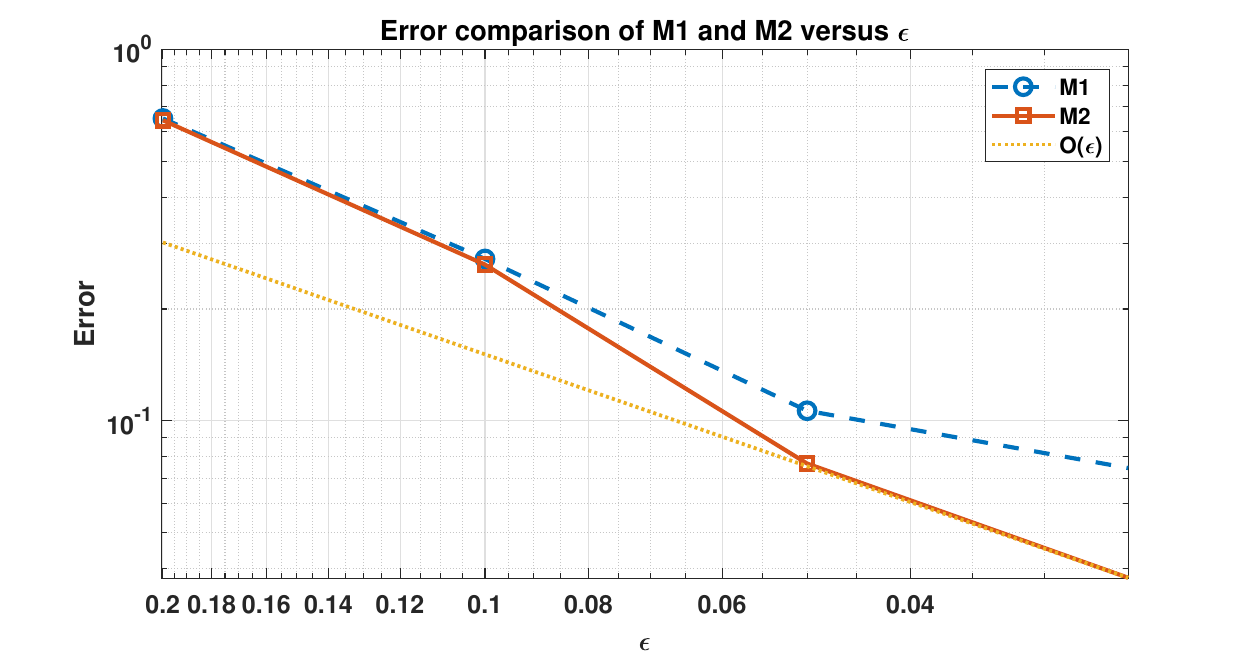}
	\caption{Final-time masked errors \(E_{M1}^\varepsilon(T)\) and
		\(E_{M2}^\varepsilon(T)\) versus \(\varepsilon\).}
	\label{fig:error-vs-epsilon-final}
\end{figure}

		\section{Conclusions}
		We proposed a framework  for  nonlinear complex nonlinear partial differential equations with Ginzburg-Landau potential that are asymptotically accurate in the strong nonlinear regime. In this
		regime, the nonlinear dynamics is decomposed into a low-dimensional evolution
		of vortices and a linear outer expansion problem governed by linear Poisson equation.  This leads naturally to a
		hybrid quantum-classical strategy, where the vortex dynamics is evolved
        classically and the linear outer problem is treated quantum mechanically via the 
        Schr\"odingerization technique.

The analysis shows that this framework admits quantum acceleration in the
small-output regime, where one seeks only a few physically relevant observables
rather than a full reconstruction of the field. For the two-dimensional
nonlinear Schr\"odinger equation, the resulting complexity exhibits an {\it
exponential} improvement in the dependence on the spatial problem size, while
the dependence on the target accuracy remains essentially linear up to
polylogarithmic factors. We further show that the same principle extends to
dissipative Ginzburg--Landau vortex dynamics and to vortex filaments in
three-dimensional superconductivity. The numerical results support the validity
of the reduction and indicate that, for the test configuration considered here,
the boundary-induced \(h_{\bm a}\)-drift is weak. Overall, the present work
suggests a concrete route toward quantum algorithms for strongly nonlinear
complex scalar field equations, including nonlinear Schr\"odinger, heat, and
wave equations with Ginzburg--Landau-type nonlinearity.

        In the future we will investigate the possibility of extending
        this method to broader classes of nonlinear PDEs.

        \section*{Acknowledgments}
SJ and NL were supported by NSFC grant No. 12341104,  the Shanghai Pilot Program for Basic Research, the Shanghai Jiao Tong University 2030 Initiative and the Fundamental Research Funds for the Central Universities. SJ was also partially supported by the NSFC grant No. 92270001.
NL also acknowledges funding from the Science and Technology Program of Shanghai, China (21JC1402900), the Science and Technology Commission of Shanghai Municipality (STCSM) grant no. 24LZ1401200 (21JC1402900) and NSFC grant No.12471411.
SJ, NL and CM were supported by the  Science and Technology Innovation Key R\&D Program of Chongqing grant No. CSTB2024TIAD-STX0035.
CM was partially supported by NSFC grant No. 12501607, the Science and Technology Commission of Shanghai Municipality (No.22DZ2229014).
		\bibliographystyle{plain} 
		\bibliography{Refnonlinear1}
	\end{document}